\documentclass[envcountsame]{llncs}
\usepackage{enum_kernel}
\title{Paradigms for 
Parameterized
Enumeration\thanks{%
Supported by a Campus France/DAAD Procope grant, Campus France Projet No~28292TE, DAAD Projekt-ID 55892324.
}}

\author{
  Nadia Creignou\inst{1}\and
  Arne Meier\inst{2}\and
  Julian-Steffen M\"uller\inst{2}\and
  Johannes Schmidt\inst{3}\and
  Heribert Vollmer\inst{2}
}

\institute{
Aix-Marseille Université~%
\email{nadia.creignou@lif.univ-mrs.fr}\and
Leibniz Universit\"at Hannover~%
\email{$\{$meier,mueller,vollmer$\}$@thi.uni-hannover.de}\and
Link{\"o}ping University~%
\email{johannes.schmidt@liu.se}}

\begin{document}

\maketitle

\begin{abstract}
The aim of the paper is to examine  the  computational  complexity and algorithmics of enumeration, the task to output all solutions of a given problem, from the point of view of parameterized complexity. 
First we define formally different notions of efficient enumeration in the context of parameterized complexity. Second we  show how different algorithmic paradigms can be used in order to get parameter-efficient enumeration algorithms in a
number of examples.
These paradigms use well-known principles from the design of parameterized decision as well as  enumeration techniques, like for instance kernelization and self-reducibility.
The concept of kernelization, in particular, leads to a characterization of fixed-parameter tractable enumeration problems.
\end{abstract}

\section{Introduction}

This paper is concerned with algorithms for and complexity studies of enumeration problems, the task of generating all solutions of a given computational problem. The area of enumeration algorithms has experienced tremendous growth over the last decade. Prime applications are query answering in databases and web search engines, data mining, web mining, bioinformatics and computational linguistics. 

Parameterized complexity theory provides a framework for a refined analysis of hard algorithmic problems. It measures complexity not only in terms of the input size, but in addition in terms of a parameter. Problem instances that exhibit structural similarities will have the same or similar parameter(s). Efficiency now means that for fixed parameter, the problem is solvable with reasonable time resources. A parameterized problem is fixed-parameter tractable (in $\FPT$) if it can be solved in polynomial time for each fixed value of the parameter, where the degree of the polynomial does not depend on the parameter. Much like in the classical setting, to give evidence that certain algorithmic problems are not in $\FPT$ one shows that they are complete for superclasses of $\FPT$, like the classes in what is known as the W-hierarchy.

Our main goal is to initiate a study of enumeration from a parameterized complexity point of view and in particular to develop parameter-efficient enumeration algorithms. Preliminary steps in this direction have been undertaken by H.~Fernau \cite{fernau02}. He considers algorithms that output \emph{all} solutions of a problem to a given instance in polynomial time for each fixed value of the parameter, where, as above, the degree of the polynomial does not depend on the parameter (let us briefly call this fpt-time). We subsume problems that exhibit such an algorithm in the class $\totalFPT$. (A similar notion was studied by Damaschke \cite{damaschke06}). Algorithms like these can of course only exists for algorithmic problems that possess only relatively few solutions for an input instance. We therefore consider algorithms that exhibit a delay between the output of two different solutions of fpt-time, and we argue that this is the ``right way'' to define tractable parameterized enumeration. The corresponding 
complexity class is called $\delayFPT$.

We then study the techniques of kernelization (stemming from parameterized complexity) and self-reducibility (well-known in the design of enumeration algorithms) under the question if they can be used to obtain parameter-efficient enumeration algorithms. We study these techniques in the context of different algorithmic problems from the context of propositional satisfiability (and vertex cover, which can, of course, also be seen as a form of weighted 2-CNF satisfiability question). We obtain a number of upper and lower bounds on the enumerability of these problems.

In the next section we introduce parameterized enumeration problems and suggest four hopefully reasonable complexity classes for their study. In the following two sections we study in turn kernelization and self-reducibility, and apply them to the problems \VC, $\maxonesSAT$ and detection of strong Horn-backdoor sets. We conclude with some open questions about related algorithmic problems.

\section{Complexity Classes for Parameterized Enumeration}

Because of the amount of solutions that enumeration algorithms possibly produce, the size of their output is often much larger (e.g., exponentially larger) than the size of their input. Therefore, polynomial time complexity is not a suitable yardstick of efficiency when analyzing their performance. As it is now agreed, one is more interested in the regularity of these algorithms rather than in their total running time. For this reason, the efficiency of an enumeration algorithm is better measured by the delay between two successive outputs, see e.g., \cite{JohnsonPY88}. The same observation holds within the context of parametrized complexity and we can define parameterized complexity classes for enumeration based on this time elapsed between two successive outputs. Let us start with the formal definition of a parameterized enumeration problem.

\begin{definition}\label{def:para-enum-pb}
	A \emph{parameterized enumeration problem}  (over a finite alphabet $\Sigma$)  is a triple $E=(Q, \kappa, \Sol)$ such that
	\begin{itemize}
		\item	$Q\subseteq \Sigma^*$,
		\item $\kappa$ is a parameterization of $\Sigma^*$, that is $\kappa\colon \Sigma^*\rightarrow \N$ is a polynomial time computable function.
		\item $\Sol:\Sigma^*\rightarrow \mathcal{P} (\Sigma^*)$ is a function such that for all $x\in\Sigma^*$, $\Sol(x)$ is a finite set and $\Sol(x)\ne\emptyset$ if and only if $x\in Q$. 
	\end{itemize}
\end{definition}

If $E=(Q, \kappa, \Sol)$ is a parameterized enumeration problem over the alphabet $\Sigma$, then we call strings $x\in\Sigma^*$ instances of $E$, the number $\kappa(x)$ the corresponding parameter, and $\Sol(x)$ the set of solutions of $x$. 
%
As an example we consider the problem of enumerating all vertex covers with bounded size of a graph.
%
\enumproblem
{\allVC}
{An undirected graph $G$ and a positive integer $k$}
{$k$}
{The set of all vertex covers of $G$ of size $\le k$}
%

An \emph{enumeration algorithm} $\mathcal{A}$ for the enumeration problem $E=(Q, \kappa, \Sol)$ is an algorithm, which on the input $x$ of $E$, outputs exactly the elements of $\Sol(x)$ without duplicates, and which terminates after a finite number of steps on every input.

At first we need to fix the notion of delay for algorithms.
\begin{definition}[Delay]
 Let $E=(Q,\kappa,\Sol)$ be a parameterized enumeration problem and $\mathcal A$ an enumeration algorithm for $E$.
 Let $x\in Q$, then we say that the $i$-th delay of $\mathcal A$ is the time between outputting the $i$-th and $(i+1)$-st solutions in $\Sol(x)$. Further, we define the $0$-th delay as the \emph{precalculation time} as the time from the start of the computation to the first output statement. Analogously, the $n$-th delay, for $n=|\Sol(x)|$, is the \emph{postcalculation time} which is the time needed after the last output statement until $\mathcal A$ terminates.
\end{definition}
We are now ready to define different notions of fixed-parameter tractability for enumeration problems.
\begin{definition}
 Let $E=(Q, \kappa, \Sol)$ be a parameterized enumeration problem and $\mathcal{A}$ an enumeration algorithm for $E$.
\begin{enumerate}
 \item The algorithm $\mathcal{A}$ is a $\totalFPT$ algorithm if there exist a computable function $t\colon \N\rightarrow \N$ and a polynomial $p$  such that for every instance $x\in\Sigma^*$, $\mathcal{A}$ outputs all solutions of $\Sol(x)$ in time at most $t(\kappa(x))\cdot p(|x|)$.
 \item The algorithm $\mathcal{A}$ is a $\delayFPT$ algorithm if there exist a computable function $t\colon \N\rightarrow \N$ and a polynomial $p$ such that for every $x\in\Sigma^*$, $\mathcal{A}$ outputs all solutions of $\Sol(x)$ with delay of at most $t(\kappa(x))\cdot p(|x|)$.
\end{enumerate}
\end{definition}

Though this will not be in the focus of the present paper, we remark that, in analogy to the non-parameterized case (see \cite{CreignouOS11,Schmidt09}), one can easily adopt the definition for $\incFPT$ algorithms whose $i$th delay is at most $t(\kappa(x))\cdot p(|x|+i)$. Similarly, one gets the notion of $\outputFPT$ algorithms which is defined by a runtime of at most $t(\kappa(x))\cdot p(|x|+|\Sol(x)|)$.

\begin{definition}
The class $\totalFPT$ (resp., $\delayFPT$)
is the class of all parameterized enumeration problems that admit a $\totalFPT$ (resp., $\delayFPT$)
enumeration algorithm.  
\end{definition}

Observe that Fernau's notion of fixed parameter enumerable \cite{fernau02} is equivalent to our term of $\totalFPT$.
Obviously the existence of a $\totalFPT$ enumeration algorithm requires that for every instance $x$ the number of solution is bounded by $f(\kappa(x))\cdot p(|x|)$, which is quite restrictive. 
Nevertheless, Fernau was able to show that the problem $\minimumVC$ (where we are only interested in vertex covers of minimum cardinality) is in $\totalFPT$, but by the just given cardinality constraint, $\allVC$ is not in $\totalFPT$. In the upcoming section we will prove that $\allVC$ is in $\delayFPT$; hence we conclude:

\begin{corollary}
$\totalFPT\subsetneq \delayFPT$.
\end{corollary}

We consider that $\delayFPT$ should be regarded as the good notion of tractability for parameterized enumeration complexity.

%
%
%
%
\section{Enumeration by Kernelization}

Kernelization is one of the most successful techniques in order to design para\-meter-efficient algorithms, and actually characterizes parameter-tractable problems. Remember that kernelization consists in a pre-processing, which is a polynomial time many-one reduction of a problem to itself with the additional property that the (size of the) image is bounded in terms of the parameter of the argument (see e.g., \cite{FlumGrohe06}).

In the following we propose a definition of an \textit{enum-kernelization}, which should be seen as a pre-processing step suitable for an efficient enumeration. 

\begin{definition}\label{def:enum-kernel}
Let $(Q,\kappa, \Sol)$ be a parameterized enumeration problem over $\Sigma$. A polynomial time computable function $K\colon\Sigma^*\rightarrow\Sigma^*$ is an \emph{enum-kernelization} of $(Q,\kappa, \Sol)$ if there exist:
\begin{enumerate}
 \item a computable function $h\colon \N\rightarrow \N$ such that for all
$x\in\Sigma^*$ we have

 \centerline{ $(x\in Q\Leftrightarrow K(x)\in Q)\hbox{ and } |K(x)|\le
h(\kappa(x))$, }
\item a computable function $f\colon{\Sigma^*}^2\rightarrow\calP(\Sigma^*)$, which from a pair $(x, w)$ where $x\in Q$ and $w\in \Sol(K(x))$, computes a subset of $\Sol(x)$, such that 
\begin{enumerate}
 \item for all $w_1, w_2\in \Sol(K(x))$, $ w_1\ne w_2\Rightarrow f(x,w_1)\cap f(x, w_2)=\emptyset$,
 \item $\displaystyle \bigcup_{w\in\Sol(K(x))}f(x, w)=\Sol(x)$
 \item there exists an enumeration algorithm $\calA_f$, which on input $(x,w)$, where $x\in Q$ and $w\in \Sol(K(x))$, enumerates all solutions of $f(x,w)$ with delay $p(|x|)\cdot t(\kappa(x))$, where $p$ is a polynomial and $t$ is a computable function.
\end{enumerate}
\end{enumerate}
If $K$ is an enum-kernelization of $(Q,\kappa, \Sol)$, then for every instance $x$ of $Q$ the image $K(x)$ is called an \emph{enum-kernel} of $x$ (under $K$).
\end{definition}

An enum-kernelization is a reduction $K$ from a parameterized enumeration problem to itself. As in the decision setting it has the property that the image is bounded in terms of the parameter argument.
For a problem instance $x$, $K(x)$ is the kernel of $x$. Observe that if $K$ is an enum-kernelization of the enumeration problem $(Q,\kappa, \Sol)$, then it is also a kernelization for the associated decision problem.
In order to fit for enumeration problems, enum-kernelizations have the additional property that the set of solutions of the original instance $x$ can be rebuilt from the set of solutions of the image $K(x)$ with $\delayFPT$. This can be seen as a generalization of the notion of \emph{full kernel} from \cite{damaschke06}, appearing in the context of what is called  \emph{subset minimization problems}. A full kernel is a kernel that contains all minimal solutions, since they represent in a certain way all solutions. 
In the context of backdoor sets (see the next section), what is known as a loss-free kernel \cite{sasz08} is a similar notion.
In our definition, an enum-kernel is a kernel that represents all solutions in the sense that they can be obtained with $\FPT$ delay from the solutions for the kernel.
\medskip

Vertex cover is a very famous problem whose parameterized complexity has been extensively studied. It is a standard example when it comes to kernelization. Let us examine it in the light of the notion of enum-kernelization.

\begin{proposition}\label{prop:allvc_enumkernel}	 
$\allVC$ has an enum-kernelization.
\end{proposition}

\begin{proof} Given a graph $G=(V,E)$ and a positive integer $k$, we are interested in enumerating all vertex covers of $G$ of size at most $k$. We prove that the famous Buss' kernelization \cite[pp.~208ff]{FlumGrohe06} provides an enum-kernelization. 
Let us remember that Buss' algorithm consists in applying repeatedly the following rules until no more reduction can be made:
\begin{enumerate}
 \item If $v$ is a vertex of degree greater than $k$, remove $v$ from the graph and decrease $k$ by one.
 \item If $v$ is an isolated vertex, remove it.
\end{enumerate}
 The algorithm terminates and the kernel $K(G)$ is the reduced graph $(V_K, E_K)$ so obtained if it has less than $k^2$ edges, and the complete graph ${\cal K}_{k+1}$ otherwise.

One verifies that whenever in a certain step of the removing process rule (1) is applicable to a vertex $v$, and $v$ is not removed immediately, then rule (1) remains applicable to $v$ also in any further step, until it is removed. Therefore, whenever we have a choice during the removal process, our choice does not influence the finally obtained graph: the kernel is unique.

Suppose that $K(G)=(V_K, E_K)$. Let $V_D$ be the set of vertices (of large degree) that are removed by the rule (1) and $V_I$ the set of vertices (isolated) that are removed by the rule (2). On the one hand every vertex cover of size $\le k$ of $G$ has to contain $V_D$. On the other hand, no vertex from $V_I$ is part of a minimal vertex cover.
Thus, all vertex covers of $G$ are obtained in considering all the vertex covers of $K(G)$, completing them by $V_D$ and by some vertices of $V_I$ up to the cardinality $k$.
Therefore, given $W$ a vertex cover of $K(G)$, then we define $f(G, W)=\{ W\cup V_D\cup V'\mid V'\subseteq V_I, |V'| \le k-|W|-|V_D|\}$. It is then clear that for $W_1\ne W_2$, $W_1,W_2 \in \Sol(K(G))$, we have that $f(G,W_1)\cap f(G, W_2)=\emptyset$. From the discussion above we have that $\bigcup_{W\in\Sol(K(G))}f(G, W)$ is the set of all $\le k$-vertex covers of $G$. Finally, given $W$ a vertex cover of $K(G)$, after a polynomial time pre-processing of $G$ by Buss's kernelization in order to compute $V_D$ and $V_I$, the enumeration of $f(G,W)$ comes down to an enumeration of all subsets of $V_I$ of size at most $k-|W|-|V_D|$. Such an enumeration can be done with polynomial delay by standard algorithms. Therefore, the set $f(G,W)$ can be enumerated with polynomial delay and, \textit{a fortiori}, with $\delayFPT$. 
\qed
\end{proof}

As in the context of decision problems, enum-kernelization actually characterizes the class of enumeration problems having $\delayFPT$-algorithm, as shown in the following theorem.

\begin{theorem} \label{thm:enum_kernel_characterization}
For every parameterized enumeration problem $(Q,\kappa, \Sol)$ over $\Sigma$, the following are equivalent:
\begin{enumerate}
 \item $(Q,\kappa, \Sol)$ is in $\delayFPT$
 \item For all $x\in \Sigma^*$ the set $\Sol(x)$ is computable and $(Q,\kappa,\Sol)$ has an enum-kernelization.
\end{enumerate}
\end{theorem}

\begin{proof}
 \begin{description}
\item[$(2)\Rightarrow (1)$:]  Let $K$ be an enum-kernelization of $(Q,\kappa, \Sol)$. Given an instance $x\in\Sigma^*$ the following algorithm enumerates all solution in $\Sol(x)$ with $\delayFPT$: compute $K(x)$ in polynomial time, say $p'(|x|)$.
Compute $\Sol(K(x))$, this requires a time $g(\kappa(x))$ for some function $g$ since the size of $K(x)$ is bounded in terms of the parameter argument. Apply successively the enumeration algorithm $\calA_f$ to the input $(x,w)$ for each $w\in\Sol(K(x))$. Since $\calA_f$ requires a delay $p(|x|)\cdot t(\kappa(x))$, the delay of this enumeration algorithm is bounded from above by $(p'(|x|)+p(|x|))\cdot(g(\kappa(x))+t(\kappa(x)))$. The correctness of the algorithm follows from the definition of an enum-kernelization (Item 2.(a) ensures that there is no repetition, Item 2.(b) that all solutions are output). 
%
%
\item[$(1)\Rightarrow (2)$:] Let $\calA$ be an enumeration algorithm for $(Q,\kappa, \Sol)$ that requires delay $p(n)\cdot t(k)$ where $p$ is a polynomial and $t$ some computable function. 
Without loss of generality we assume that $p(n)\ge n$ for all positive integer $n$. If $Q=\emptyset$ or $Q=\Sigma^*$ then  $(Q,\kappa, \Sol)$ has a trivial kernelization that maps every $x\in\Sigma^*$ to the empty string $\epsilon$. 
If $Q=\emptyset$ we are done. If $Q=\Sigma^*$, then fix $w_\epsilon\in\Sol(\epsilon)$ and set for all $x$, $f(x,w_\epsilon)=\Sol(x)$ and $f(x,w)=\emptyset$ for $w\in\Sol(\epsilon)\setminus\{w_\epsilon\}$.
Otherwise, we fix $x_0\in\Sigma^*\setminus Q$, and $x_1\in Q$ with $w_1\in\Sol(x_1)$. 

\medskip

The following algorithm $\calA'$ computes an enum-kernelization for $(Q,\kappa, \Sol)$: Given $x\in\Sigma^*$ with $n:=|x|$ and $k=\kappa(x)$, 
\begin{enumerate}
\item  the algorithm simulates $p(n)\cdot p(n)$ steps of $\calA$. 
\item  If it stops with the answer ``no solution'', then set $K(x)=x_0$ (since $x_0\notin Q$, the function $f$ does not need to be defined). 
\item  If a solution is output within this time, then set $K(x)=x_1$, $f(x, w_1)=\Sol(x)$ and $f(x,w)=\emptyset$ for all $w\in\Sol(x_1)\setminus \{w_1\}$. 
\item  If it does not output a solution within this time, then it holds $n\leq p(n)\leq t(k)$ and then we set $K(x)=x$, and $f(x,w)=\{w\}$ for all $w\in\Sol(x)$. 
\end{enumerate}

Clearly $K(x)$ can thus be computed in time $p(n)^2$, $|K(x)|\le |x_0|+|x_1|+t(k)$, $(x\in Q\Leftrightarrow K(x)\in Q)$, and the function $f$ we have obtained satisfies all the requirements of \Cref{def:enum-kernel}, in particular the enumeration algorithm $\calA$ can be used to enumerate $f(x, w)$ when applicable. Therefore $K$ provides indeed an enum-kernelization for $(Q,\kappa, \Sol)$.\qed
\end{description}
\end{proof}

\begin{corollary}
$\allVC$ is in $\delayFPT$.
\end{corollary}

\begin{remark}
Observe that in the proof of \Cref{prop:allvc_enumkernel}, the enumeration of the sets of solutions obtained from a solution $W$ of $K(G)$ is enumerable even with polynomial-delay, we do not need fpt delay. We will show in the full paper that this is a general property: Enum-kernelization can be equivalently defined as \FPT-preprocessing followed by enumeration with polynomial delay.
\end{remark}

\section{Enumeration by Self-Reducibility}

In this section we would like to exemplify the use of the algorithmic paradigm of self-reducibility (\cite{Schnorr76,KhullerV91,Schmidt09}), on which various enumeration algorithms are based in the literature. The self-reducibility property of a problem allows a ``search-reduces-to-decision'' algorithm to enumerate the solutions. This technique seems quite appropriate for satisfiability related problems. We will first investigate the enumeration of models of a formula having weight at least $k$, and then turn to strong \HORN-backdoor sets of size $k$. In the first example the  underlying decision problem can be solved in using kernelization (see \cite{KratschMW10}), while in the second it is solved in using the bounded-search-tree technique.

\subsection{Enumeration classification for $\maxonesSAT$}\label{subsec:classification}

The self-reducibility technique was in particular applied in order to enumerate all satisfying assignments of a generalized CNF-formula \cite{CreignouH97}, thus allowing to identify classes of formulas which admit efficient enumeration algorithms. In the context of parameterized complexity a natural problem is $\maxonesSAT$, in which the question is to decide whether there exists a satisfying assignment of weight at least $k$, the integer $k$ being the parameter. 
We are here interested in the corresponding enumeration problem, and we will study it for generalized CNF formulas, namely in Schaefer's framework. In order to state the problem we are interested in more formally, we need some notation.

A \emph{logical relation} of arity $k$ is a relation $R\subseteq\{0,1\}^k$. By abuse of notation we do not make a difference between a relation and its predicate symbol.
A \emph{constraint}, $C$, is a formula  $C=R(x_1,\dots,x_k)$, where $R$ is a logical relation of arity $k$ and the $x_i$'s are (not necessarily distinct) variables. If $u$ and $v$ are two variables, then $C[u/v]$ denotes the constraint obtained from $C$ in replacing each occurrence of $v$ by $u$. 
An assignment $m$ of truth values to the variables \emph{satisfies} the constraint $C$ if $\bigl(m(x_1),\dots,m(x_k)\bigr)\in R$. A \emph{constraint language} $\Gamma$ is a finite set of logical relations. A \emph{$\Gamma$-formula} $\phi$, is a conjunction of constraints using only logical relations from $\Gamma$ and is hence a quantifier-free first order formula. With $\var(\phi)$ we denote the set of variables appearing in $\phi$. A $\Gamma$-formula $\phi$ is satisfied by an assignment $m:\var(\phi)\to\{0,1\}$ if $m$ satisfies all constraints in $\phi$ simultaneously (such a satisfying assignment is also called a \emph{model} of $\phi$). The \textit{weight of a model} is given by the number of variables set to true. 
Assuming a canonical order on the variables we can regard models as tuples in the obvious way and we do not distinguish between a formula $\phi$ and the logical relation $R_\phi$ it defines, i.e., the relation consisting of all models of $\phi$. In the following we will consider two particular constraints, namely $\imp(x,y)=(x\rightarrow y)$ and $\true(x)=(x)$.

We are interested in the following parameterized enumeration problem.

\enumproblem%
{$\enummaxonesSAT(\Gamma)$}%
{A $\Gamma$-formula $\varphi$ and a positive integer $k$}%
{$k$}%
{All assignments satisfying $\varphi$ of weight $\ge k$}%

The corresponding decision problem, denoted by $\maxonesSAT(\Gamma)$, i.e., the problem to decide if a given formula has a satisfying assignment of a given weight, has been studied by Kratsch et al.~\cite{KratschMW10}. They completely settle the question of its parameterized complexity in Schaefer's framework. To state their result we need some terminology concerning types of Boolean relations.
%
%

Well known already from Schaefer's original paper \cite{Schaefer78} are the following seven classes:
We say that a Boolean relation~$R$ is \emph{$a$-valid} (for $a\in\{0,1\}$) if $R(a,\ldots, a)=1$. 
A relation~$R$ is \emph{Horn} (resp., \emph{dual Horn}) if $R$ can be defined by a CNF formula which is Horn (resp., dual Horn), i.e., every clause contains at most one positive (resp., negative) literal. 
A relation~$R$ is \emph{bijunctive} if $R$ can be defined by a 2-CNF formula.
A relation $R$ is \emph{affine} if it can be defined by an \emph{affine} formula, i.e., conjunctions of XOR-clauses (consisting of an XOR of some variables plus maybe the constant 1)---such a formula may also be seen as a system of linear equations over GF$[2]$. 
A relation $R$  is \emph{complementive} if for all $m\in R$ we have also $\vec{1} \oplus m\in R$. 

Kratsch et al. \cite{KratschMW10} introduce a new restriction of the class of bijunctive relations as follows.
For this they use the notion of \textit{frozen implementation}, stemming from \cite{Nordhz09}. Let $\varphi$ be a formula and $x\in\var(\varphi)$, then $x$ is said to be \textit{frozen} in $\varphi$ if it is assigned the same truth value in all its models. Further, we say that $\Gamma$ \textit{freezingly implements} a given relation $R$ if there is a $\Gamma$-formula $\varphi$ such that $R(x_1,\ldots x_n)\equiv\exists X\varphi$, where $\varphi$ uses variables from $X\cup \{x_1,\ldots x_n\}$ only, and all variables in $X$ are frozen in $\varphi$. For sake of readability, we denote by $\clos{\Gamma}$ the set of all relations that can be freezingly implemented by $\Gamma$. 
A relation~$R$ is \emph{\stronglybijunctive} if it is in $\clos{\{(x\lor y), (x\ne y), (x\rightarrow y)\}}$.

Finally, we say that a constraint language $\Gamma$ has one of the just defined properties 
if every relation in $\Gamma$ has the property.

\begin{theorem}\label{thm:classification_maxones_decision} \cite[Thm.~7]{KratschMW10}
If $\Gamma$ is 1-valid, dual-Horn, affine, or \stronglybijunctive, then $\maxonesSAT(\Gamma)$ is in  $\FPT$.  Otherwise $\maxonesSAT(\Gamma)$ is $\Wone$-hard.
\end{theorem}

Interestingly we can get a complete classification for enumeration as well. The fixed-parameter efficient enumeration algorithms are obtained through the algorithmic paradigm of self-reducibility. 

We would like to mention that an analogously defined decision problem $\minonesSAT(\Gamma)$
is in $\FPT$ (by a bounded search-tree algorithm) and the enumeration problem has $\FPT$-delay for all constraint langauges $\Gamma$. The decision problem $\exactonesSAT(\Gamma)$ has been studied by Marx \cite{marx05} and shown to be in $\FPT$ iff $\Gamma$ has a property called ``weakly separable''. We remark that it can be shown, again by making use of self-reducibility, that under the same conditions, the corresponding enumeration algorithm has $\FPT$-delay. This will be presented in the full paper. In the present submission we concentrate on the, as we think, more interesting maximization problem, since here, the classification of the complexity of the enumeration problem differs from the one for the decision problem, as we state in the following theorem.

\begin{theorem}\label{thm:classification_maxones}  
If $\Gamma$ is  dual-Horn, affine, or \stronglybijunctive, then there is a $\delayFPT$ algorithm for $\enummaxonesSAT(\Gamma)$. Otherwise such an algorithm does not exist unless $\Wone=\FPT$.
\end{theorem}
%
%

It would be interesting for those cases of $\Gamma$ that do not admit a $\delayFPT$ algorithm to determine an upper bound besides the trivial exponential time bound to enumerate all solutions. 
In particular, are there such sets $\Gamma$ for which $\enummaxonesSAT(\Gamma)$ is in $\outputFPT$?


\begin{proof} (of \Cref{thm:classification_maxones})
We first propose a canonical algorithm for enumerating all satisfying assignments of weight at least $k$. 
The function {\verb'HasMaxOnes'}$(\phi,k)$ tests if the formula $\phi$ has a model of weight at least $k$.

\begin{algorithm}
	
\KwIn{A formula $\phi$ with $\var(\phi)=\{x_1,\ldots, x_n \}$, an integer $k$}
\KwOut{All sat.\ assignments (given as sets of variables) of $\phi$ of weight $\ge k$.}
	\SetKwFunction{maxones}{HasMaxOnes}
	\SetKwFunction{generate}{Generate}
	\SetKwFunction{algo}{algo}
	\SetKwFunction{proc}{proc}
	\LinesNumbered\setcounter{AlgoLine}{0}
	\nl\lIf{$\maxones(\phi,k)$}{$\generate(\phi, \emptyset, k, n)$}
	
	\BlankLine\BlankLine\BlankLine
	\proce $\generate(\phi, M, w, p):$
	\BlankLine \setcounter{AlgoLine}{0}
	
	\nl\lIf{$w=0$ or $p=0$}{\Return $M$}

	\nl\Else{
	\nl	\If{$\maxones(\phi[x_p=1], w-1)$}{\nl$\generate(\phi[x_p=1], M\cup\{x_p\},w-1,p-1)$}
	\nl	\lIf{$\maxones(\phi[x_p=0], w)$}{$\generate(\phi[x_p=0], M, w,p-1)$}
	}
  \caption{Algorithm with procedure}

\label[algorithm]{algo:generate-sat}
\caption{Enumerate the models of weight at least $k$}
\end{algorithm}
%


Observe that if $\Gamma$ is dual-Horn, affine, or \stronglybijunctive, then according to \Cref{thm:classification_maxones} the procedure {\verb'HasMaxOnes'}$(\phi,k)$ can be performed in $\FPT$. 
%
%
Moreover essentially if $\phi$ is dual-Horn (resp., affine, \stronglybijunctive) then so are $\phi[x_p=0]$ and $\phi[x_p=1]$ for any variable $x_p$. Therefore, in all these cases the proposed enumeration algorithm has clearly $\delayFPT$.
Now it remains to deal with the hard cases. Roughly speaking we will show that in these cases either finding one solution or finding two solutions is hard, thus excluding the existence of an efficient enumeration algorithm. Let us consider the problem $\maxonesSAT^*(\Gamma)$, which given a formula $\phi$ and an integer $k$ consists in deciding whether $\phi$ has a nontrivial (i.e., non-all-1) model of weight at least $k$. We will show that when $\Gamma$ is neither dual-Horn, nor affine, nor \stronglybijunctive, then $\maxonesSAT^*(\Gamma)$ is either $\Wone$-hard or $\NP$-hard for $k=0$.
This implies that if there is a $\delayFPT$ algorithm that enumerates all models of weight at least $k$ of a $\Gamma$-formula, then $\FPT=\Wone$ or even, in the second case, $\PTime=\NP$, hence the claim of our theorem will follow.

We now proceed to proving hardness of $\maxonesSAT^*(\Gamma)$.

If $\maxonesSAT(\Gamma)$ is $\Wone$-hard, then obviously so is $\maxonesSAT^*(\Gamma)$.
Therefore, according to \Cref{thm:classification_maxones} it remains to consider the case where $\Gamma$ is 1-valid but neither dual-Horn, nor affine, nor \stronglybijunctive. In this case $\maxonesSAT(\Gamma)$ is trivial, whereas $\maxonesSAT(\Gamma\cup\{0\})$ is hard. We will use the following fact: 
\begin{align}\label{eqn:red_maxones_maxonesstar}
 \maxonesSAT(\Gamma\cup\{0\})\redfpt\maxonesSAT^*(\Gamma\cup\{\imp\}).
\end{align}
The proof of this claim is easy. Given a $\Gamma\cup\{0\}$-formula $\varphi$ over the set of variables $\{x_1,\ldots, x_n\}$, let us consider the $\Gamma\cup\{\imp\}$-formula defined as $\varphi':=\varphi[f/0]\land \bigwedge_{i=1}^n \imp(f, x_i)$ where $f$ is a fresh variable. It is easy to see that there is a one-to-one correspondence between the models of $\varphi$ and those of $\varphi'$ that set $f$ to $0$, moreover the only model of $\varphi'$ that sets $f$ to 1 is the all-$1$ assignment. Therefore, $\varphi$ has a model of weight at least $k$ if and only if $\varphi'$ has one nontrivial model of weight at least $k$, thus proving the claim.

%
%
Making use of the above defined notion of freezing implementations, we obtain a possibility to get rid of the relation $\imp$ in (\ref{eqn:red_maxones_maxonesstar}):
\begin{align}\label{eqn:red_maxones_R_gamma}
\hbox{ If } R\in\clos\Gamma, \hbox{ then }\maxonesSAT^*(R)\redfpt\maxonesSAT^*(\Gamma).
\end{align}
Indeed the frozen implementation gives us a procedure to transform any $R$-formula into a satisfiability equivalent $\Gamma$-formula with existentially quantified variables. The fact that the implementation ``freezes'' the existentially quantified variables makes it possible to remove the quantifiers, while preserving the information on the weight of the  solutions. Thus, in order to prove that $\maxonesSAT^*(\Gamma)$ is hard we will have essentially two strategies:
\begin{itemize}
 \item either we exhibit a relation $R\in\clos\Gamma$ such that $\maxonesSAT^*(R)$ is hard and then we conclude thanks to
(\ref{eqn:red_maxones_R_gamma}),
 \item or we prove that $\imp\in\clos\Gamma$ and then we conclude thanks to (\ref{eqn:red_maxones_maxonesstar}) since $\maxonesSAT(\Gamma\cup\{0\})$ is hard. (In the case $\Gamma$ is complementive we use a symmetric version of implication $\symimp(x,y,z)=(z=0\land \imp(x,y))\lor (z=1\land \imp(y,x))$).
\end{itemize}
The rest of the proof consists in finding relevant implementations in a very standard way (see e.g., \cite{crkhsu01}), therefore we here give only a sketch. 
%
%
%
Suppose for instance that $\Gamma$ is not 0-valid (the other cases can
 be dealt with in a similar manner).
 It is easy to show that
 $\true\in\clos\Gamma$. Let us first consider 
   $R\in\Gamma$ a non-dual-Horn relation of arity $m$. Consider the
 constraint $C=R(x_1,\ldots, x_m)$. Since $R$ is non-dual-Horn there
 exist  $m_1$ and $m_2$ in $R$ such that
  $m_1\lor m_2\notin R$. For $i,j\in\{0,1\}$, set $\displaystyle V_{i,j}=\{x\mid
 x\in V,
 m_1(x)=i\land m_2(x)=j\}$. Consider the $\{R\}$-constraint:
 $M(x,y,z,t)=C[x/V_{0,0},\, y/V_{0,1},\,
 z/V_{1,0},\, t/V_{1,1}].$ 
 Now, let  $R'\in\Gamma$ a non-affine relation of arity $m'$. Consider the
 constraint $C=R(x_1,\ldots, x_{m'})$. Since $R'$ is non-affine and 1-valid there
 exist  $m'_1$ and $m'_2$ in $R'$ such that
  $(m'_1\oplus m'_2\oplus (1,\ldots, 1)\notin R$. For $i,j\in\{0,1\}$, set
 $\displaystyle V_{i,j}=\{x\mid
 x\in V,
 m'_1(x)=i\land m'_2(x)=j\}$. Consider the $\{R'\}$-constraint:
 $M'(x,y,z,t)=C[x/V_{0,0},\, y/V_{0,1},\,
 z/V_{1,0},\, t/V_{1,1}].$  Finally consider the ternary relation $Q$ defined by
 $Q(x,y,z)=\exists t M(x,y,z,t)\land M'(x,y,z,t)\land \true(t)$. Clearly
 $Q\in\clos\Gamma$. Moreover, by construction the relation $Q$ contains the
 tuples $0011$, $0101$ and $1111$, and does   contain neither  $0111$ (because of
 the constraint $M$), nor $1001$ (because of $M'$).  There are therefore three
  tuples for which we do not know whether they belong to $Q$ or not, and this
 makes 8 cases to investigate. It is easy to check that $\imp\in\clos Q$, and
 hence $\imp\in\clos\Gamma$, in all cases except when $Q=\{001, 010, 111\}$ or
 $Q=\{001, 010, 111, 110\}$. For the six cases such that $\imp\in\clos\Gamma$ we
 conclude with (\ref{eqn:red_maxones_maxonesstar}). In the two remaining cases,
 it is easy  to verify  that $\satstar(Q)$ is $\NP$-hard in using
 \cite{CreignouH97} ($Q$ is non Schaefer), thus we deduce that
 $\maxonesstarSAT(Q)$ is $\NP$-hard for $k=0$, and hence we conclude with
 (\ref{eqn:red_maxones_R_gamma}).
\qed
\end{proof}
A more detailed version of his proof will be included in the journal version of the paper.

\subsection{Enumeration of strong \HORN-backdoor sets}
\label{subsec:sbds}

We consider here the enumeration of strong backdoor sets. Let us introduce some relevant terminology \cite{wigose03}. Consider a formula $\phi$, a set $V$ of variables of $\phi$, $V\subseteq \var(\phi)$. For a truth assignment $\tau$, $\phi(\tau)$ denotes the result of removing all clauses from $\phi$ which contain a literal $x$ with $\tau(x) = 1$ and removing literals $y$ with $\tau(y) = 0$ from the remaining clauses. 

The set $V$ is a \textit{strong \HORN-backdoor set} of $\phi$ if for all truth assignment $\tau\colon V\rightarrow\{0,1\}$ we have $\phi(\tau)\in\HORN$. 
Observe that equivalently $V$ is a strong \HORN-backdoor set of $\phi$ if $\phi|_V$ is \HORN, where $\phi|_V$ denotes the formula obtained from $\phi$ in deleting in $\phi$ all occurrences of variables from $V$. 

Now let us consider the following enumeration problem.
\enumproblem
{\exactSBDS{$\HORN$}}
{A  formula $\phi$ in CNF}
{$k$}
{The set of all strong \HORN-backdoor sets of $\phi$ of size exactly $k$}

From \cite{nirasz04} we know that detection of strong \HORN-backdoor sets is in $\FPT$. 
In using a variant of bounded-search tree the authors use in their $\FPT$-algorithm, together with self-reducibility we get an efficient enumeration algorithm for all strong \HORN-backdoor sets of size $k$.

\begin{theorem} \label{prop:sbds_delayFPT}
\exactSBDS{$\HORN$} is in $\delayFPT$.
\end{theorem}

\begin{proof}
The procedure {\tt GenerateSBDS}$(\phi,B, k, V)$ depicted in \Cref{algo:generate-sat} enumerates all sets $S\subseteq V$ of size $k$ such that $B\cup S$ is a strong \HORN-backdoor set for $\phi$, while the function {\tt Exists-SBDS}$(\phi, k, V)$ tests if $\phi$ has a strong \HORN-backdoor set of size exactly $k$ made of variables from $V$.

\begin{algorithm}[ht!]
\KwIn{A formula $\phi$, an integer $k$}
\KwOut{All strong \HORN-backdoor sets of size $k$.}
\SetKwFunction{existsSBDS}{Exists-SBDS}
\SetKwFunction{generateSBDS}{GenerateSBDS}
\nl\lIf{$\existsSBDS(\phi, k, \var(\phi))$}{$\generateSBDS(\phi, \emptyset, k, \var(\phi))$}

\BlankLine\BlankLine\BlankLine
\proce $\generateSBDS(\phi, B, k, V):$
\BlankLine \setcounter{AlgoLine}{0}

\nl\lIf{$k=0$ or $V=\emptyset$}
 {
 	\Return $B$
 }
 
\nl\Else
{
	\nl\If{$\existsSBDS(\phi|_{B\cup\{\min(V)\}}, k-1,V\setminus\{\min(V)\})$}
 	{
 		\nl$\generateSBDS(\phi, B\cup\{\min(V)\}, k-1,
V\setminus\{\min(V)\})$
 	}
 	\nl\If{$\existsSBDS(\phi|_B, k, V\setminus\{\min(V)\})$}
 	{
 		\nl$\generateSBDS(\phi, B, k, V\setminus\{\min(V)\})$
 	}
}

\BlankLine\BlankLine\BlankLine
\func $\existsSBDS(\phi, k, V):$
\BlankLine \setcounter{AlgoLine}{0}

\nl\If{$k=0$ or $V=\emptyset$}
{
	\nl\lIf{$\phi|_V\in\HORN$}
	{
		\Return true
	}
	\lElse
	{
		\Return false
	}
}
\nl\uIf{there is a clause $C$ with two  positive literals $p_1,p_2$}
{
	\nl\uIf{exactly one of $p_1$ and $p_2$ is in $V$, say $p_1\in V, p_2\notin V$}
	{
		\nl\lIf{$\existsSBDS(\phi|_{\{p_1\}}, k-1, V\setminus\{p_1\})$}
		{
			\Return true
		}
	}
	\nl\Else
	{
		\nl\If{$p_1\in V$ and  $p_2\in V$}
		{
			\nl\lIf{$\existsSBDS(\phi|_{\{p_1\}},k-1, V\setminus\{p_1\})$}
			{
				\Return true
			}
			
			\nl\lIf{$\existsSBDS(\phi|_{\{p_2\}}, k-1, V\setminus\{p_2\})$}
			{
				\Return true
			}
		}
	}
	{
		\nl\Return false 
	}
}
\nl\lElse{\Return true}
\label[algorithm]{algo:generate-hbd}
\caption{Enumerate all strong \HORN-backdoor sets of size $k$}
\end{algorithm}

 

 The point that this algorithm is indeed in $\delayFPT$ relies on the fact that the function {\tt Exists-SBDS} depicted in \Cref{algo:generate-hbd} is in $\FPT$. This function is an adaptation of the one proposed in \cite{nirasz04}.
There Nishimura et al.\ use an important fact holding for non-\HORN clauses (i.e., clauses contains at least two positive literals): if $p_1,p_2$ are two positive literals then either one of them must belong to any strong backdoor set of the complete formula.
 
 In their algorithm they just go through all clauses for these occurrences. However for our task, the enumeration of the backdoor sets, it is very important to take care of the ordering of variables. The reason for this is the following. Using the algorithm without changes makes it impossible to enumerate the backdoor sets because wrong sets would be considered: e.g., for some formula $\phi$ and variables $x_1,\dots,x_n$ let $B=\{x_2,x_4,x_5\}$ be the only strong backdoor set. Then, during the enumeration process, one would come to the point where the sets with $x_2$ have been investigated (our algorithm just enumerates from the smallest variable index to the highest). When we start investigating the sets containing $x_4$, the procedure would then wrongly say "yes there is a backdoor set containing $x_4$" which is not desired in this situation because we finished considering $x_2$ (and only want to investigate backdoor sets that do not contain $x_2$).
 
 Therefore the algorithm needs to consider only the variables in the set $V$ where in each recursive call the minimum  variable (i.e., the one with smallest index) is removed from the set $V$ of considered variables.
\qed
%

%
%
\end{proof}

\section{Conclusion}

We made a first step to develop a computational complexity theory for parameterized enumeration problems by defining a number of, as we hope, useful complexity classes.
We examined two design paradigms for parameterized algorithms from the point of view of enumeration. Thus we obtained a number of upper bounds and also some lower bounds for important algorithmic problems, mainly from the area of propositional satisfiability. 

 As further promising problems we consider the cluster editing problem \cite{damaschke06} and the $k$-flip-SAT problem \cite{szeider11}.

Of course it will be very interesting to examine further algorithmic paradigms for their suitability to obtain enumeration algorithms. 
Here, we think of the technique of bounded search trees and the use of structural graph properties like treewidth.\bigskip

\noindent\textbf{Acknowledgements.}\\
We are very thankful to Fr\'ed\'eric Olive for helpful discussions. We also acknowledge many helpful comments from the reviewers.

\bibliographystyle{plain}
\bibliography{enum_kernel}
\end{document}